\documentclass[12pt,onecolumn,draftcls]{IEEEtran} 

\IEEEoverridecommandlockouts                              
\usepackage{amssymb}
\usepackage{amsmath}
\usepackage{color}
\usepackage{epsfig}

\newtheorem{remark}{Remark}
\newtheorem{theorem}{Theorem}[section]

\begin{document}
\title{
\bf
Optical Flow Sensing and the Inverse Perception Problem for Flying Bats
}
\author{ Zhaodan Kong,  Kayhan \"Ozcimder, Nathan Fuller, Alison Greco, Diane Theriault, Zheng Wu, Thomas Kunz, Margrit Betke,  John Baillieul\\[0.1in]
\thanks{Research support by the Office of Naval Research through the ODDR\&E MURI10 Program Grant Number N00014-10-1-0952 as well as  the U.S.\ National Science Foundation HCC Grant Number IIS-071322 is gratefully acknowledged.}
\thanks{Kong, \"Ozcimder, and Baillieul are with the College of Engineering; Fuller,  Greco, and Kunz are with the Department of Biology; Theriault, Wu, and Betke are with the Department of Computer Science; all at Boston University.}
\thanks{Corresponding author is Baillieul: {\tt johnb@bu.edu}}}
\IEEEaftertitletext{\vspace{-2\baselineskip}} 
\maketitle

\vspace{0.1in}
\begin{abstract}
The movements of birds, bats, and other flying species are governed by complex sensorimotor systems that allow the animals to react to stationary environmental features as well as to wind disturbances, other animals in nearby airspace, and a wide variety of unexpected challenges.   The paper and talk will describe research that analyzes the three-dimensional trajectories of bats flying in a habitat in Texas.  The trajectories are computed with stereoscopic methods using data from synchronous thermal videos that were recorded with high temporal and spatial resolution from three viewpoints.  Following our previously reported work, we examine the possibility that bat trajectories in this habitat are governed by optical flow sensing that interpolates periodic distance measurements from echolocation.  Using an idealized geometry of bat eyes, we introduce the concept of {\em time-to-transit}, and recall some research that suggests that this quantity is computed by the animals' visual cortex.  Several steering control laws based on time-to-transit are proposed for an idealized flight model, and it is shown that these can be used to replicate the observed flight of what we identify as {\em typical bats}. Although the vision-based motion control laws we propose and the protocols for switching between them are quite simple, some of the trajectories that have been synthesized are qualitatively bat-like.  Examination of the control protocols that generate these trajectories suggests that bat motions are governed both by their reactions to a subset of key feature points as well by their memories of where these feature points are located. \end{abstract}

\section{Intorduction}\setcounter{equation}{0}
\label{sec:jb:Intro}

This paper describes research aimed at using observed flight trajectories of a species of bats to understand the ways in which the animals use sensory perceptions of their environments to control their motions.  The work is based on data recovered from a large collection of 3-dimensional video records of {\em Myotis velifer}, emerging from a cave on the Bamberger Ranch Preserve in Johnson City, Texas.  These are cave-roosting bats that live in southern North America and Central America and are a large bodied species of the {\em Myotis} genus, weighing about 14 grams and having a wingspan of 30 cm (\cite{Farney},\cite{Fitch}).

Bats perceive features within their environment using complex combinations of sensory organs.  Many species, including the {\em M.~velfer} studied in the present paper, are able to perceive distances to objects by means of echolocation.  It is likely, however, that vision also plays a role in guiding the these animals as they fly.  Similar to other Yangochiroptera (microbats), \textit{M.~velifer} have relatively small eyes that are principally directed sidewards from opposite sides of the head.  Their retinas are rod-dominated, making them well suited to their nocturnal niche. In addition, their retinas contain dense horizontal connections but few vertical connections, which suggests that they are specialized to detect motions and contours under nocturnal illumination at the expense of high feature discrimination acuity~\cite{neuweiler2000biology}. However, contrary to traditional belief that bats possess a simplistic visual system, recent evidence suggests that several bat species, including {\em M. velifer}, have functional S opsin genes~\cite{muller2009bat}. It is unclear whether these genes are expressed in {\em M. velifer} retinae, but the presence of functional genes suggests that {\em M. velifer} may be able to see UV light and thus possess mesopic vision that is effective at dusk and dawn and on brightly moonlit nights~\cite{Wang}. The optic nerves of \textit{M.~velifer}'s left and right eyes remain separate. Each optic nerve crosses over completely to the contralateral side of the brain~\cite{neuweiler2000biology}. This anatomical characteristic suggests that, at least at the lower level, information from each eye is processed separately by the brain. 

The research presented in this paper is part of a growing body of work aimed at understanding the sensorimotor motor control of flying animals.  We refer to \cite{Justh},\cite{Sebesta},\cite{Alaeddini}, and \cite{Boardman} for examples. 
In the following sections, simple models of vision-guided horizontal flight control are proposed and evaluated in terms of their potential for explaining the observed motions of {\em M.~velifer} bats.  The organization is as follows.  Section \ref{sec:jb:BatData} describes the processing of video data and the computer rendering of observed bat trajectories.  Section \ref{sec:jb:motion} discusses the concept of {\em time-to-transit}, how this can be determined from the flow of feature images on the retinas of animals in flight, and the emerging evidence that animals use perceptions of time-to-transit to guide their flight behaviors.  Section \ref{sec:jb:distanceMaintain} proposes what we have called {\em single feature} and {\em paired feature} control laws that are based on optical flow and that can be used to steer a simple model vehicle in a way that  has the potential for replicating  reconstructed bat trajectories.  Section \ref{sec:jb:Simulations} presents motion simulations under various assumed environmental sensing protocols that exhibit characteristics of bat flight, and Section \ref{sec:jb:Discussion} draws some conclusions about what can be learned from detailed observation of animal flight behaviors coupled with flight control simulations.  We also discuss the need for further research to understand broader classes of vision-based control as well as the use of hybrid sensing whereby vision and echolocation are used synergistically.

\section{Recovery of flight paths from 3-dimensional video records}
\label{sec:jb:BatData}

Using high-resolution thermal video recordings of {\em M. velifer} emerging from a cave on the Bamberger Ranch in Texas,  three-dimensional reconstructions of 405 different trajectories were created [16, 19]. Errors inevitably appear in these reconstructions due to uncertainties arising from bats flying outside the 3D calibration region, from  occasional occlusions, and from misidentifying homologous points on the bat's body in the three views, especially when its size in at least one view is small.  Smoothing and filtering were carried out along the lines discussed in \cite{Krishna} but in this case using cubic spline smoothing $\hat\mu(\cdot)$ with a smoothing factor $\lambda=0.85$.  (See \cite{Stat2}.)  
\begin{equation}
\min_{\hat\mu}\left\{\lambda \sum_{i=1}^{n}(Y_i-\hat{\mu}(x_i))^2+(1-\lambda) \int_{x_1}^{x_n}\hat{\mu}''(x)^2dx\right\}.
\end{equation}
Over the range of the parameter, $0\le \lambda\le1$, $\lambda =0 $ corresponds to a linear least squares fit to the data, while $\lambda =1$ corresponds to a cubic spline interpolation passing through every data point.  The parameter $\lambda$ is chosen such that the smoothing is good enough for noise cancellation without loosing too much information.   Since the goal is to understand the aggregate flight characteristics of the bats, the smoothed trajectories were parameterized by arc length and truncated to have a common length from the point at which they entered the field of view to the point at which the image resolution was overly noisy because of its distance from the cameras.  The smoothed and reparameterized trajectories were then filtered to remove those that exhibited anomalies such as leaving the field of view prematurely.   Finally, over the 254 trajectory segments that were retained for study, it was noted that the bats descended at a fairly  steady rate so as to follow the descending slope of a hill.  These trajectories appeared to be largely confined to a plane, with only small deviations above or below.  Hence our initial attempt to understand how the bats' sensory perceptions were guiding their movement has been focused on models of motion control in this plane.

The planar projection of the smoothed and filtered trajectories is shown in Fig.~\ref{fig:Kayhan:Trajectories}.  The cameras that recorded the flights were located just outside and to the left of the rectangular region displayed in the figure.  The bats entered the field of view of the three video cameras from the left.  At the top of the figure, the small triangles correspond to trees in a wooded area, and at the center of the figure there are two significant features labeled {\em vine} and {\em pole}.  The vine is a natural feature that runs from the ground up to a fairly high tree branch so that the bats must fly either to the left or the right of it. The pole was placed  as a vertical marker to calibrate the cameras, and its height was such that bats could either fly over or around it. 

\begin{figure}[h]
\begin{center}
\includegraphics[scale=0.46]{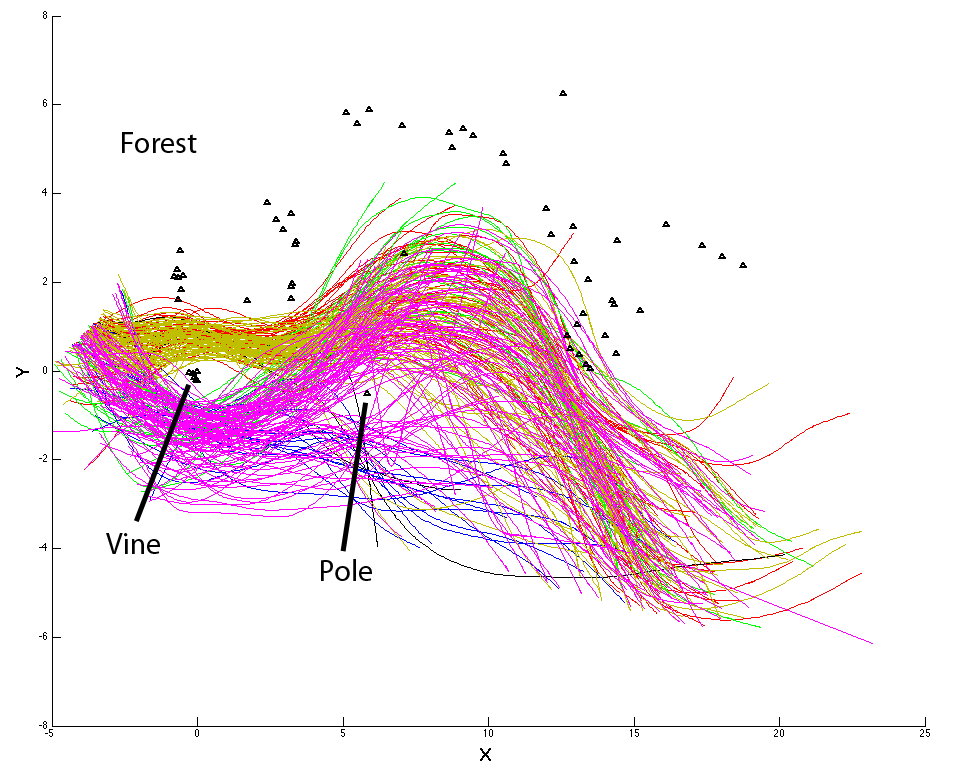}
\end{center}
\caption{Smoothed and projected trajectories. In this schematic rendering, the bats emerging from the cave enter on the left and fly to the right until they are effectively out of range.  The trajectories are rendered in different colors that reflect where the bats flew with respect to the vine (which they had to pass on the left or the right) and the pole (which they could fly around or over).  The triangles are visual features in the woods---mainly points on tree branches.}
\label{fig:Kayhan:Trajectories}
\end{figure}

We initially subdivided the 254 trajectories into six groups depending on whether they went around the vine on the left or right together with
whether they flew to the left, right or over the pole.  Visual inspection suggested that the differences among the trajectories with respect to the pole were insignificant.  Thus we divided the set of trajectories into two classes: the first comprised of 115 bats who flew to the left of the vine and the second being the 139 bats who flew to the right.  Having parameterized all trajectories by arc length, we adopted the viewpoint that we could recreate the flight path of a ``typical bat'' by computing mean trajectories within each class.  Fig.~\ref{fig:jb:BambergerWoods} below displays the mean path of the 115 bats who flew to the left of the vine  in red, and the mean path of the 139 bats who flew to the right of the vine in blue.  The green curve was synthesized by a flight control law that will be discussed in Sections \ref{sec:jb:distanceMaintain} and \ref{sec:jb:Simulations}.  The green flyer reacts more strongly to perceived environmental features than a typical bat would.  It is interesting to note that although the red, blue, and green paths all deviate from one another, they rejoin at the end of the flight domain.
To justify averaging as a way to create a typical trajectory, more analysis of the path statistics, including the consistency of speed profiles along the paths, is needed.  This analysis has been carried out, but it is beyond the scope of this short paper and will appear elsewhere. 


\section{Review of motion control using optical sensor feedback}
\label{sec:jb:motion}

Sebesta and Baillieul \cite{Sebesta} showed how the well-known optical parameter $\tau$ could be used to guide the motion of a moving optical sensor.  It is worth redoing this analysis with a special emphasis on the eye geometry of the {\em M.~velifer}.  In \cite{Lee-Reddish} and \cite{Sebesta}, $\tau$ was described in terms of a geometric picture that is most appropriate for  animals (like humans) that have a forward-facing field of view.  The eyes of the {\em M.~velifer} are shifted toward the sides of the head, greatly expanding their field of view and, at the same time, giving them higher acuity in resolving objects that are off to one side or the other (\cite{Jones}).   Because of the placement of their eyes and due to neural connectivity patterns of their photoreceptors, these bats have enhanced ability to orient themselves with respect to features in the lateral visual field.  We conjecture that the bats may use optical flow sensing of features in the lateral field to guide their motions.  In order to explore this, we introduce the concept of {\em time-to-transit} and discuss how time-to-transit is easily determined from the movement of feature images on an animal's retina or on the image plane of a camera.  Consider 
\begin{figure}[ht] 
\begin{center}\includegraphics[width=2.7in]{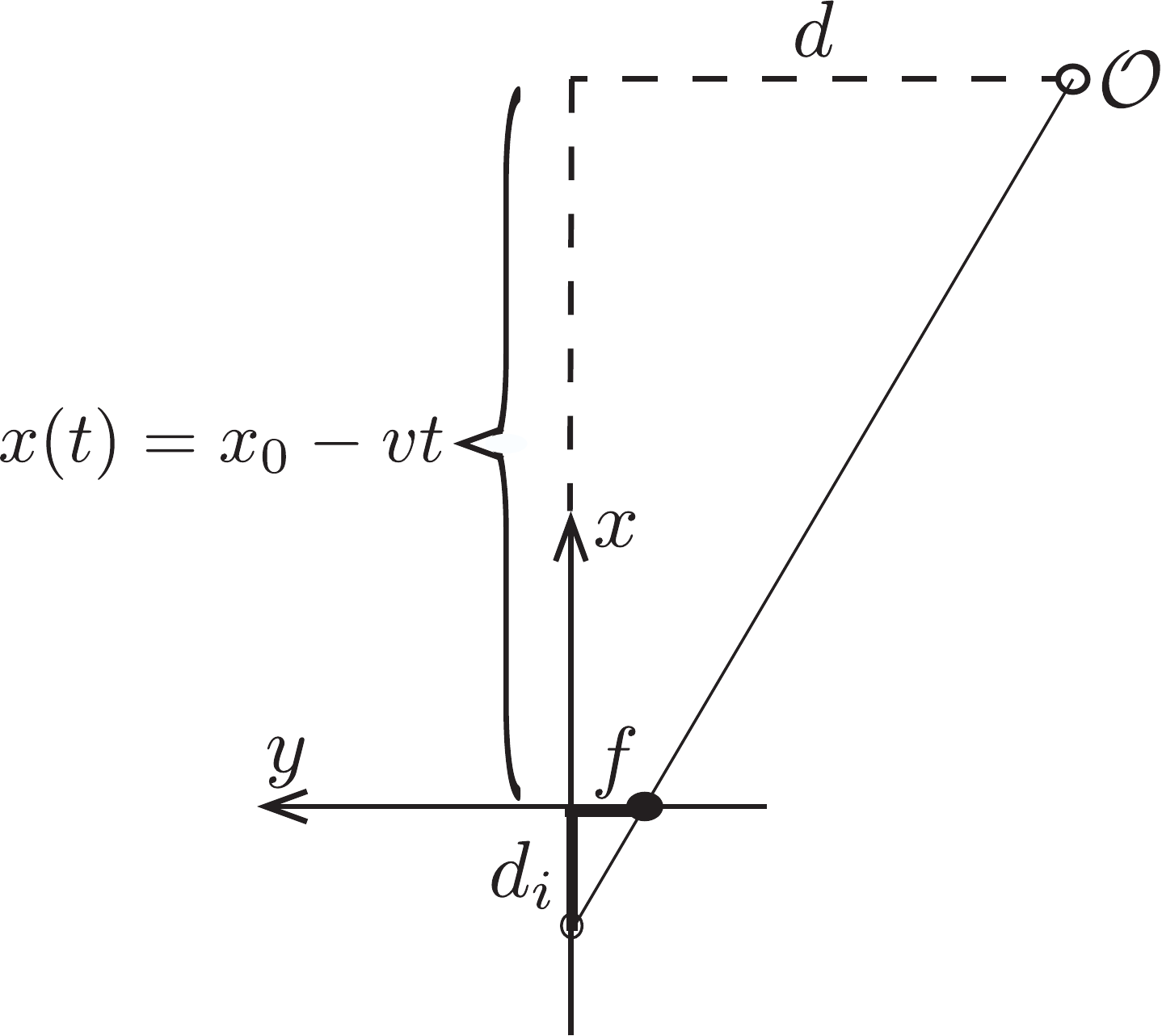}\end{center}
\caption{Optical flow of feature images for a sideward-looking imaging system.} 
\label{fig:jb:OpticalFlow} 
\end{figure}  
the idealized planar vehicle depicted in Fig.\ \ref{fig:jb:OpticalFlow}.  (Since we are interested in planar rendering of observed animal trajectories (per Section \ref{sec:jb:BatData}), models of motion in the plane are appropriate.)  The direction of motion is aligned with the vehicle body frame $x$-axis, and the feature ${\cal O}$ is observed with a pinhole camera system whose camera axis is aligned with the negative body frame $y$-axis.  The image of the point feature ${\cal O}$ is at $d_i$, a negative quantity as it is depicted in the figure.

We suppose the vehicle moves in the direction of its body frame $x$-axis at a constant speed $v$.  If the motion is initiated at a point $x_0$ along the line of flight at time $t=0$, it will cross a line that is perpendicular to the line of flight and passes through the feature point ${\cal O}$ at time $t=x_0/v$.  This quantity is called the {\em time-to-transit}, and we denote it by $\tau$.  We note that as the figure is drawn, the image point corresponding to the feature in our idealized camera lies at $d_i$ in the body frame $x$-axis, and the focal point lies at $-f$ on the body frame $y$-axis.  ($f$ is the camera focal length.)  It is clear that the similarity of triangles implies the relationship
\[
\frac{d}{x(t)-d_i(t)}=-\frac{f}{d_i(t)},
\]
and from this it follows that
\[
\frac{d_i}{\dot d_i}=\frac{x_0}{v}-t.
\]
This is zero when $t=x_0/v$ (the time at which the vehicle crosses the line perpendicular to its path and passing through ${\cal O}$).   At $t=0$, we see that $d_i/\dot d_i = x_0/v =\tau$ is the time-to-transit. The general conclusion is that if $d_i(t)$ is the location of an image feature, $\tau(t)=d_i(t)/\dot d_i(t)$ is the time remaining until the camera is directly abeam of the actual feature, provided that the speed and heading are held constant. 

This simple argument is very similar to the way in which $\tau$ was introduced in \cite{Lee-Reddish} and \cite{Sebesta}, but with an extremely important difference.  The analysis in \cite{Lee-Reddish} and \cite{Sebesta} was carried out for a forward-looking imaging system, whereas here we have treated a side looking imaging system.  For forward-looking imaging systems, when the vehicle or animal transits the line to the feature point, the image point will have disappeared from the field of view---going outside the range of peripheral vision.  For side looking eyes (cameras), however, the images of features being transited will be located near the camera focal point or at the {\em fovea centralis}---the point on the retina where visual acuity is greatest.  This suggests that animals like bats will have an acute awareness of objects that are to one side of their flight paths.

Terms such as ``looming'' and ``time-to-collision'' have appeared infrequently in the literature of mobile robot control, with notable exceptions being \cite{Moshtagh1}, and \cite{Sebesta}.  Nevertheless, there is a growing literature dealing with biological sensorimotor control in which vision based sensing of the time-to-transit key feature points is seen as playing a central role in regulating motion behaviors (\cite{Kaiser},\cite{Lee-Reddish},\cite{Wang1}).

\section{Distance maintenance using optical flow}
\label{sec:jb:distanceMaintain}

In \cite{Sebesta}, we showed that the difference between times-to-transit for two separated features could be used as a feedback signal to guide a vehicle along a path between them.  We now turn our attention to the problem of using optical flow to steer a vehicle along a path lying a short distance outside the boundary of an array of obstacles.  The question we would like to answer is whether paths that resemble our reconstructed bat motions can be generated.
\begin{figure}[ht] 
\begin{center}\includegraphics[width=2.7in]{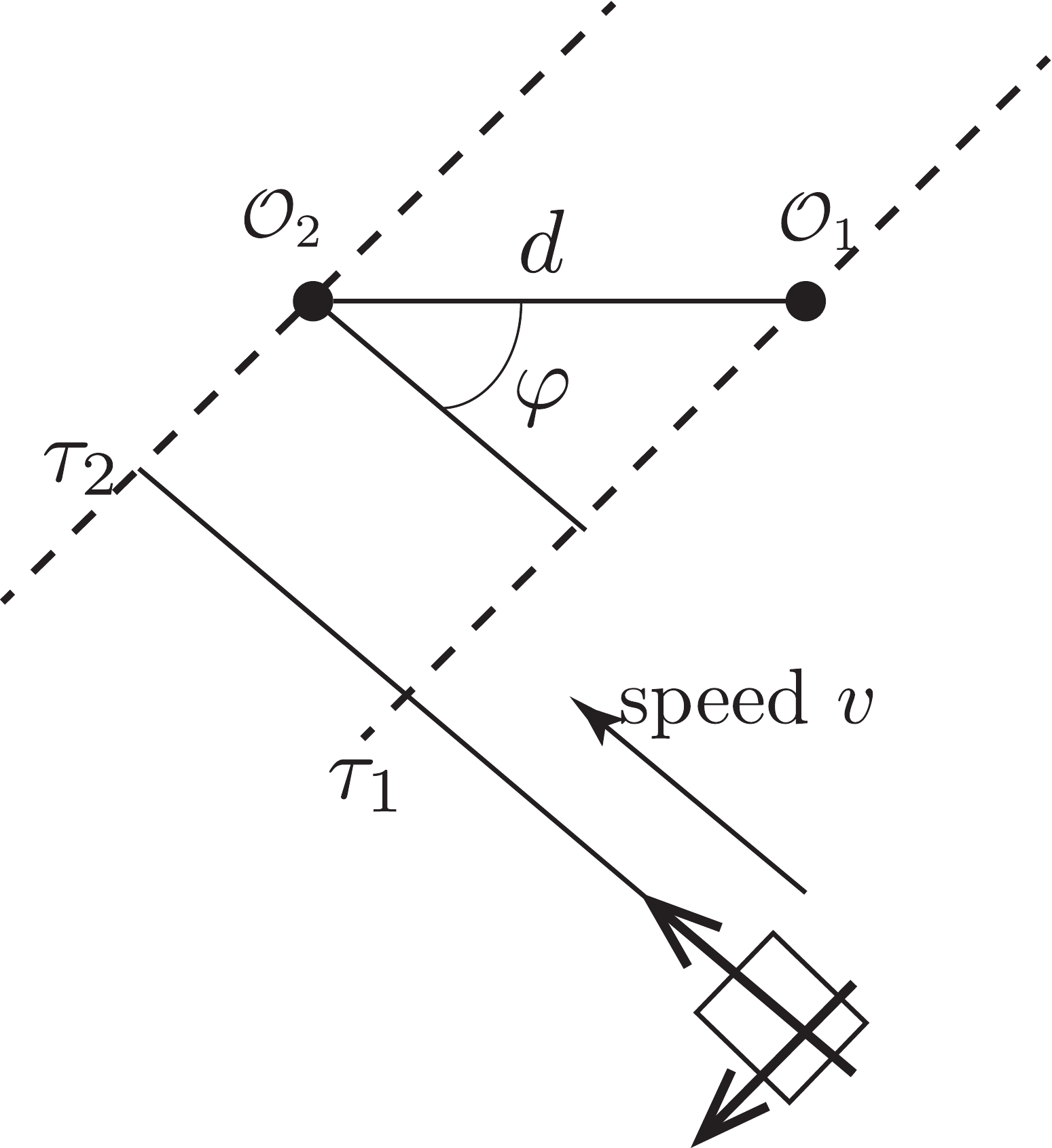}\end{center}
\caption{The difference between $\tau_2$ and $\tau_1$ is maximized if the line of travel is parallel to the line segment joining ${\cal O}_1$ and ${\cal O}_2$.} 
\label{fig:jb:MaxTransit} 
\end{figure}  
Consider the vehicle motion at constant speed $v$ that is depicted in Fig.\ \ref{fig:jb:MaxTransit}.  The features ${\cal O}_1$ and ${\cal O}_2$ are $d$ units of distance from one another and register in the vehicle's side looking imaging system that was described in Section \ref{sec:jb:motion}.  It is clear that the difference between the transit times $\tau_2$ and $\tau_1$ is given by $d\cos\varphi /v$, and this difference will be maximized when $\varphi=0$.  In other words, if the movement can be controlled in a way that ensures that $\tau_2-\tau_1$ is maximized, the direction of travel will be aligned with the line segment from ${\cal O}_1$ to ${\cal O}_2$.

It is well established that animals are good at following sensory gradients (\cite{Grunbaum},\cite{Stocker}).  Because of this, we have adopted a working conjecture that the bat trajectories discussed in Section \ref{sec:jb:BatData} follow the edge of a wooded area by means of maximizing differences in time-to-transit of features along the boundary.  (See Fig.~\ref{fig:Kayhan:Trajectories}.)  To prove that this approach to navigation is theoretically feasible, we note that apart from our observation that the time-to-transit can be determined entirely from the optical flow that is sensed by a side-looking imaging system, it is also noted as in \cite{Sebesta} that time-to-transit is a purely geometric quantity.  Referring to Fig.~\ref{fig:jb:OpticalFlow}, we adopt a simple kinematic model of planar motion
\begin{equation}
\left(\begin{array}{c}
\dot x \\
\dot y \\
\dot\theta\end{array}\right) = \left(\begin{array}{l}
v\cos\theta \\
v\sin\theta \\
u\end{array}\right),
\label{eq:jb:BasicVehicle}
\end{equation}
where $v$ is the forward speed in the direction of the body-frame $x$-axis, and $u$ is the turning rate.  By carrying out an analysis that differs from \cite{Sebesta} only in minor details, we find that the time-to-transit a feature located at $(x_w,y_w)$ in the world frame by a vehicle traveling at constant speed $v$ and having configuration $(x,y,\theta)$ (in world frame coordinates) is
\begin{equation}
\tau=\frac{\cos\theta\ (x_w-x) + \sin\theta\ (y_w-y)}{v}.
\label{eq:jb:tau}
\end{equation}
Assuming the feature point $(x_w,y_w)$ is fixed in the world frame and letting only $\theta$ vary, we find that $\tau=\tau(\theta)$ is maximized when the vehicle is aimed directly at the feature point.  On the other hand, for two separated features ${\cal O}_1,{\cal O}_2$ located at $(x_1,y_1),(x_2,y_2)$ respectively, the heading that maximizes the difference in times-to-transit, $\tau_2(\theta)-\tau_1(\theta)$, is clearly aligned (as noted above) with the direction $(x_2-x_1,y_2-y_1)$, and at this heading we have $\tau_2{'}(\theta)-\tau_1{'}(\theta)=0$.

\begin{theorem}
Consider point features ${\cal O}_1,{\cal O}_2$ located respectively at $(x_1,y_1)$ and $(x_2,y_2)$.  Let $\tau_j(t)$ be the time-to-transit associated with feature ${\cal O}_j$ for $j=1,2$.  Suppose the initial orientation, $\theta_0$,  of the vehicle is such that $\tau_2>\tau_1$ (which implies that $\cos\theta\,(x_2-x_1)+\sin\theta\,(y_2-y_1)>0$). Further assume that the vehicle travels at constant speed $v=1$.  Then for any $k>0$, the steering control
\begin{equation}
u=u(t)=k\, [\tau_2{'}(\theta(t))-\tau_1{'}(\theta(t))],
\label{eq:jb:control}
\end{equation}
where $\tau_j^{'}(\theta)=\frac{\partial \tau_j}{\partial\theta}$, will asymptotically align the vehicle with the line segment directed from ${\cal O}_1$ to ${\cal O}_2$ \label{thm:jb:Steering}
\end{theorem}
\begin{proof}
Let $(v_x,v_y)$ designate the planar direction from feature ${\cal O}_1$ to ${\cal O}_2$:
\[
\left( \begin{array}{c} v_x \\ v_y \end{array}\right) = \left( \begin{array}{c} x_2-x_1 \\ y_2-y_1 \end{array}\right)\frac{1}{\sqrt{(x_2-x_1)^2+(y_2-y_1)^2}}.
\]
The function $V(\theta)=1-\cos\theta\,v_x-\sin\theta\,v_y$ will serve as a Lyapunov function with Lie derivative
\[
\begin{array}{ccl}{\cal L}_uV(\theta(t))&=&\frac{\partial V}{\partial\theta}\cdot\dot\theta\\[0.1in]
& = & -k\rho\,(-\sin\theta\,v_x+\cos\theta\,v_y)^2,\end{array}
\]
where $\rho=\sqrt{(x_2-x_1)^2+(y_2-y_1)^2}$.  The set $V(\theta)\le 1$ is compact and invariant under the motion, and by LaSalle's theorem the motion evolves asymptotically toward the set
\[
\{\theta\,:\,V(\theta)\le 1\}\cap \{\theta\,:\,{\cal L}_uV(\theta)=0\}.
\]
The unique value of $\theta$, $0\le\theta<2\pi$, lying in this set is such that $(\cos\theta,\sin\theta)=(v_x,v_y)$.  This specifies that the asymptotic direction of the motion is aligned with the  line from ${\cal O}_1$ to ${\cal O}_2$ as stated in the theorem.
\end{proof}

\begin{remark}
The theorem is conservative in the sense that the hypotheses are intended to restrict the initial conditions to configurations in which the moving vehicle has a non-zero component of its motion in the direction of the line from ${\cal O}_1$ to ${\cal O}_2$.   If the vehicle had access to its orientation $\theta$ in space and to the world-frame coordinates of the features $(\,(x_1,y_1),(x_2,y_2)\,)$, then the control law (\ref{eq:jb:control}) could be written $u(t)=k\,[-\sin\theta\,(x_2-x_1)+\cos\theta\,(y_2-y_1)]$.  Assuming the system has access this global configuration information, the control law will steer the vehicle to alignment with the feature from every initial configuration except those in which the initial $\theta_0$ has the vehicle aligned with the direction $(x_1-x_2,y_1-y_2)$ (i.e.\ aligned in exactly the opposite direction from the goal alignment of the theorem).  The existence of this singular direction is a familiar characteristic of kinematic control laws for vehicle models of the form (\ref{eq:jb:BasicVehicle}), \cite{Baillieul03}.
\end{remark}

\begin{remark}
We note that the value of $\tau$ that is associated with a visible feature as specified in (\ref{eq:jb:tau}) is a purely geometric quantity.  It will be negative in the case that the feature lies behind the vehicle on its current line of flight.  For side looking eyes, it could still be visible, and would be perceived as being both negative and becoming increasingly negative as the motion continues.  The important point is that the control law (\ref{eq:jb:control}) does not assume that either feature lies ahead of the vehicle on its current heading.
\end{remark}

\begin{remark}
The theorem states what is achieved by keeping the difference $\tau_2-\tau_1$ at its maximum value.  How flying animals might detect this maximum is not understood, but we speculate that small magnitude saccadic eye movements might be used as an energy efficient means of detecting what corrections to the vehicle's heading are needed to keep $\tau_2-\tau_1$ at its maximum value.
\end{remark}

\begin{remark}
({\em Model validity})  Field observations together with the unsmoothed reconstructed flight data make clear that the relatively smooth trajectories that are produced by the model (\ref{eq:jb:BasicVehicle}) do not capture the continual rapid, short-distance lateral, pitching, and rolling motions that make the animal flight movements anything but smooth.  It is our working assumption that these high-frequency deviations from smoothness can be thought of as noise that can be ignored in our initial attempt to synthesize vision-based bat-like trajectories.
\end{remark}

\begin{figure}[ht] 
\begin{center}\includegraphics[width=3.5in]{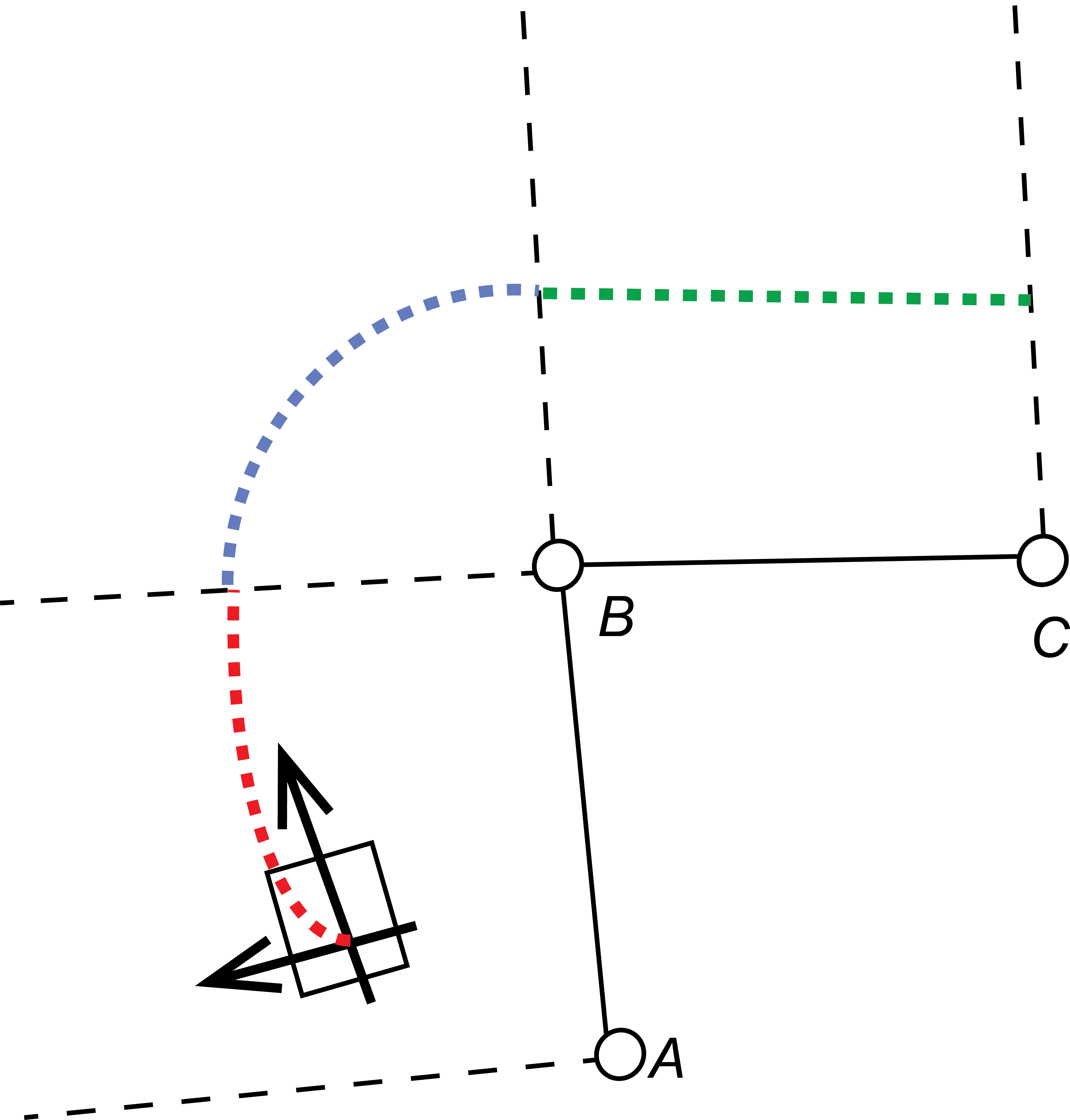}\end{center}
\caption{Optical flow sensing can be used to guide the idealized vehicle to maintain a constant distance (dashed line) from three noncolinear obstacles $A, B$, and $C$. Here there are three trajectory segments: the maximized difference of time-to-transit is used to align the vehicle with features $A$ and $B$ on the first (red line); then time-to-transit with respect to feature $B$ is held constant (blue), and finally, the control law returns to maximizing the difference of time-to-transit to align the trajectory with the line from $B$ to $C$ (green).}
\label{fig:jb:Triangle} 
\end{figure}

\subsection{Distance-keeping to a curved row of objects}

To study the curvature of trajectories prescribed by (\ref{eq:jb:BasicVehicle}), we normalize the problem such that $(x_1,y_1)=(0,0),\,(x_2,y_2)=(1,0)$ and $v=1$.  The steering control is then written $u(t)=-k\sin\theta(t)$.  Note that this is just the curvature.  It follows as a corollary of Theorem  \ref{thm:jb:Steering} and Remark 1 that for $\theta_0\ne \pi$, the absolute value $|u(t)|$ is a monotonically decreasing function that asymptotically approaches zero.  The rate of decrease is determined by the gain parameter $k$, and for any threshold $\alpha>0$, we can explicitly compute the time $T_{\alpha}(k)$ at which $|u|$ becomes less than $\alpha$.  This is made precise in the following.

\begin{theorem}
Consider the trajectory (\ref{eq:jb:BasicVehicle}) prescribed by $v\equiv 1$ and the steering law
\begin{equation}
u(t)=-k\,\sin\theta(t).
\label{eq:jb:control1}
\end{equation}
Let  $k>0$ and $0<\alpha<k|\sin\theta_0|$.
Then $|u(t)|$ is a monotonically decreasing function that takes on the value $\alpha$ at time $t=T_{\alpha}(k)$, where
\[
T_{\alpha}(k)=\frac{1}k\left[\log \left(\frac{\theta_0 }{2}\right)-\log \left(\tan \left(\frac{1}{2} \sin
   ^{-1}\left(\frac{\alpha}{k}\right)\right)\right)\right].
\]
\end{theorem}

\begin{proof}
First we note that the steering equation
\[
\dot\theta=-k\sin\theta,\ \theta(0)=\theta_0
\]
can be integrated explicitly in closed form to give
\[
\theta(t)=2 \tan^{-1}\left(\tan(\frac{\theta_0}{2})\,e^{-kt}\right).
\]
The signed curvature of the trajectory generated by (\ref{eq:jb:BasicVehicle}) is given by
\[
\kappa(k,t)=u(t)=-k\,\sin\left[2\tan^{-1}\left(\tan(\frac{\theta_0}{2})\,e^{-kt}\right)\right].
\]
For each $k$, as noted above, $\kappa(k,t)$ is a monotonically decreasing function of $t$ so that for each $\alpha$ with $0<\alpha<k\,|\sin(\theta_0)|$, we have a unique solution $t>0$ to the equation $\alpha=\kappa(k,t)$.  Elementary but slightly tedious algebra yields this solution as
\[
t_{\alpha}=T_{\alpha}(k)=\frac{1}k\left[\log \left(\frac{\theta_0 }{2}\right)-\log \left(\tan \left(\frac{1}{2} \sin
   ^{-1}\left(\frac{\alpha}{k}\right)\right)\right)\right],
\]
proving the theorem.
\end{proof}
The theorem is useful in understanding the time required for the control law (\ref{eq:jb:control}) to align the flight path with a given straight line direction.  Another potentially useful control law that can be based on sensed optical flow is one that keeps the time-to-transit fixed at zero.  Consider three features forming the vertices of a triangle as depicted in Fig.~\ref{fig:jb:Triangle}.  The vehicle (\ref{eq:jb:BasicVehicle}) can be steered around the three features by using the following control protocol.  First, assuming the vehicle starts at or below the dashed line segment through feature $A$, apply  control law law (\ref{eq:jb:control}) (more precisely, the control  $u(t)=k\, [\tau_{B}{'}(\theta(t))-\tau_{A}{'}(\theta(t))]$) so that it aligns itself with the segment $AB$.  As soon as the condition $\tau_B=0$ is met (i.e. when the vehicle transits the dashed line through $B$ that is perpendicular to its line of travel), the vehicle switches to a control to keep $\tau_B(t)\equiv 0$.  This will cause the vehicle to execute a circular arc that maintains a constant distance from feature $B$.  This control continues to be applied until the difference $\tau_C(t)-\tau_B(t)$ attains a maximum value, at which time the vehicle switches to the control law $u(t)=k\, [\tau_{C}{'}(\theta(t))-\tau_{B}{'}(\theta(t))]$.  It will continue in this way in a direction parallel to the segment $BC$ at the same distance to the side of the segment.  In terms of the notation of Table 1, the path segmentation is prescribed by $u_d[A,B;t]\rightarrow u_c[B;t] \rightarrow u_d[B,C;t]$, with the notation that will be explained in the next section.

Again, it is important to emphasize that we have not attempted to incorporate neurologically based models of how bats might sense that transit time differences are at a maximum value or how they might fly so as to keep a transit time constant.  Optimum seeking control laws and control laws that steer vehicles so as to keep sensed quantities constant are well known (\cite{Baronov}), but the details of how such control can be carried out using vision and other sensing by the animals remains an open question.  In the next section, we consider a number of vision-based control strategies that give rise to flight paths resembling the bat trajectories described in Section \ref{sec:jb:BatData}.


\section{Flight simulations for understanding feature based animal navigation}
\label{sec:jb:Simulations}

A complete understanding of the sensorimotor dynamics that produce the flight patterns described in Section \ref{sec:jb:BatData} remains a distant goal.  Nevertheless, it is of interest to discover what can be learned from trying to reproduce animal-like trajectories with suitably tuned versions of the simple control laws and protocols that have been discussed in the preceding sections.  At the outset, there are two related but fundamentally different questions.  First, how closely can we come to producing a ``typical'' bat trajectory using an idealized flight vehicle with various vision-based control laws of the form we have described.  Here the term ``typical trajectory,''  refers to the mean trajectories described in Section \ref{sec:jb:BatData}.  The second question, which may be more difficult, is can we produce simple models and protocols that predict and replicate the variability among the bats in this simple setting.  This is to say, can we find control laws and protocols such that all 254 trajectory reconstructions can be reproduced by simple variations of the model parameters.  

In addressing these questions, we are considering what we call {\em single feature} and {\em paired feature} optical flow based control.  An example of the former is the circling control law (keeping time-to-transit a feature constant), and examples of the latter are the control law (\ref{eq:jb:control}) and the time-to-transit based law proposed in \cite{Sebesta} for flying between two features.  These laws may be unrealistically simple, and if it is indeed the case that animals navigate by means of optical flow and other forms of visual feedback, the numbers of features they react to at any given instant of a flight may be quite high.  As we previously reported (\cite{Sebesta}), sparse optical-flow algorithms (BRISK and FREAK) are being studied together with approaches  to extracting meaningful control actions based on the output of such algorithms.  This work is beyond the scope of the present paper.  

\vspace{0.0in}
\begin{table*}[ht]
\begin{center}
\begin{tabular}{||c|l|l||}
\hline
\multicolumn{3}{||c||}
{{\em Time-to-Transit} vision-based steering controls}\\ \hline
$u_c[{\cal O};t] $& single-feature control & keeps $\tau_{\cal O}$ constant\\
& ---feature circling & follows circular arc at\\
&&  constant radius from\\
&& feature ${\cal O}$,\\ \hline
$u_d[{\cal O}_1,{\cal O}_2;t]$ & paired-feature control & aligns with the line \\
& ---distance maintenance & segment from ${\cal O}_1$ to ${\cal O}_2$;\\
&&  see Thm.~4.1.\\ \hline
$u_p[{\cal O}_1,{\cal O}_2;t]$ & paired feature control & control law from \cite{Sebesta}\\
& --- steers between features & to steer vehicle on\\
&& path between ${\cal O}_1,{\cal O}_2$.\\ \hline
\end{tabular}\\[0.1in]
{\bf Table 1}
\end{center}
\end{table*}
We turn instead to considering whether an appropriately sequenced set of {\em single feature} and {\em paired feature} vision-based control control laws can produce an animal-like trajectory in the habitat described in Section \ref{sec:jb:BatData}. 
While the analysis of \cite{Sebesta} and Section \ref{sec:jb:distanceMaintain} above characterizes the motions produced by each control law in Table 1, the key to constructing complete flight paths will be the creation of biologically meaningful rules and protocols for switching between these laws.  For the case of {\em M.~velifer}, it is known that echolocation is the dominant distance sensing modality, and even in the case that the animal motions described in Section \ref{sec:jb:BatData} are governed primarily by vision-based steering, it is likely that some distance information from echolocation plays a role in determining the flight paths.  It is worth noting, however, that  {\em M.~velifer} may be using visual distance cues to supplement what is obtained from echolocation.  For animals with sideward-looking eyes, the time-to-transit $\tau_{\cal O}$ associated with any feature ${\cal O}$ is independent of how far off to the side of the flight path the feature lies.  Visual cues for this exist, however.  Elementary geometric consideration regarding the schematic 
of Fig.~\ref{fig:jb:OpticalFlow} show that in the final instant before transit, the image of a distant feature will lie realtively much closer to the {\em fovea centralis} (center of focus) than a nearby feature.  At this point, our research is agnostic about how the animals detect distance to features, but a detailed statistical analysis of the 254 trajectories chosen for study indicates that distance to key feature points has a significant effect on the geometry of the flight paths.

\begin{figure}[ht] 
\begin{center}\includegraphics[width=5.0in]{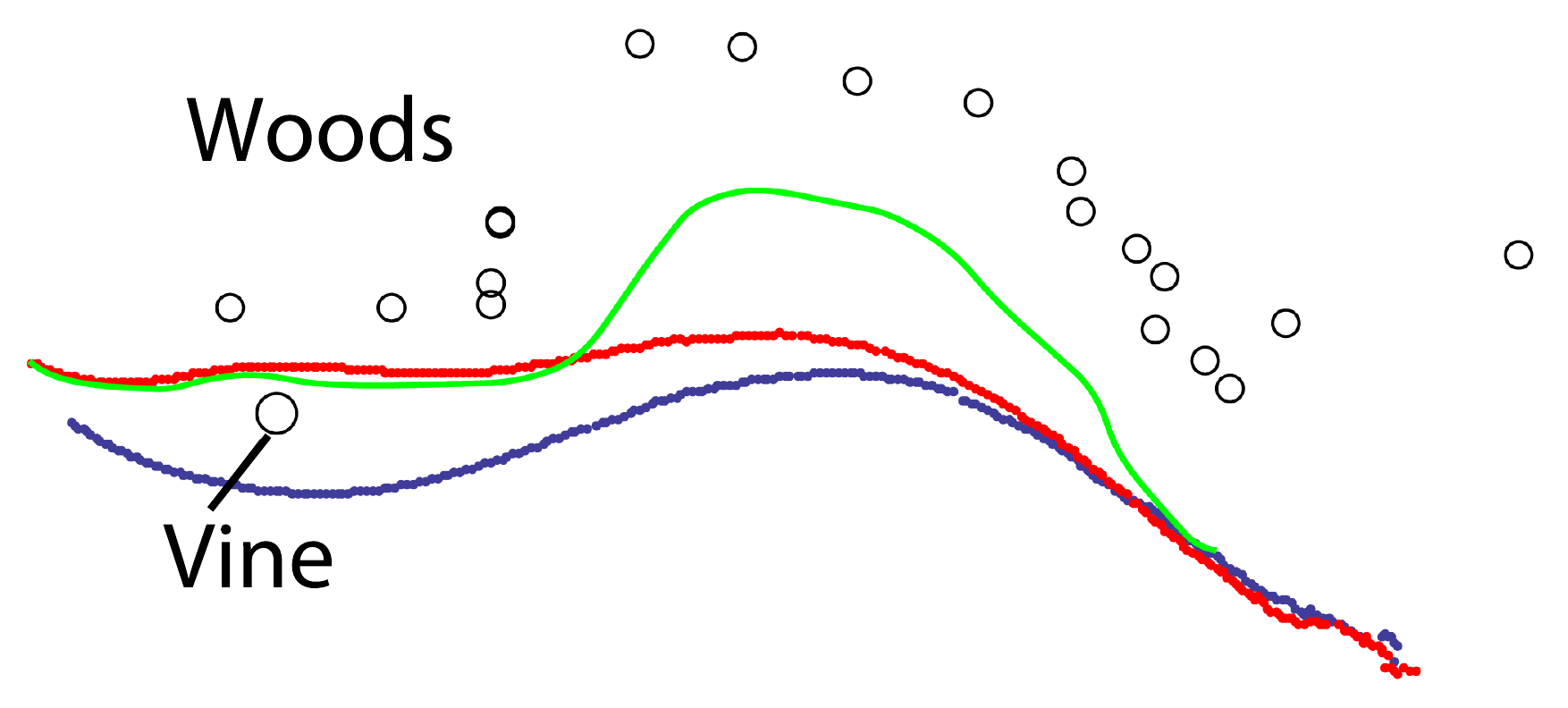}\end{center}
\caption{The tree features are those on the boundary of the wooded area depicted in Fig.~\ref{fig:Kayhan:Trajectories}. The green trajectory has been generated using the distance-keeping control described in the text.  The red curve is the average of the bat trajectories that avoided the vine obstruction on the left, while the blue curve is the average of those that passed on the right. It is interesting to note that the averaged bat trajectories do not make the same excursion into the woods, but that all trajectories come together at the end of the observed flight sequence.} 
\label{fig:jb:BambergerWoods} 
\end{figure}  

We conclude the paper by discussing the synthesis of two different trajectories using sequences of the control primitives listed in Table 1.
Biologists have been interested for some time in the question of whether navigation strategies are cue-directed or rely more on spatial memory~\cite{schnitzler2003spatial}. The cue-directed strategies employ sensed cues from the environment to choose the route of travel while the spatial memory strategies have the bats using a learned topological or even metric representation of their environment to choose their route. Evidence can be found that both strategies are employed by animal flyers. For instance, it has been found that \textit{Eptesicus fuscus} use acoustic landmarks for spatial orientation~\cite{jensen2005echolocating}. Bats can quickly locate a net opening by using a landmark even though the locations of the landmark and the net are changed (their relative positions are fixed). It has also been found that \textit{Glossophaga commissarisi} primarily use spatial memory to locate food targets even though they remember both spatial and object attributes of their target~\cite{thiele2005hierarchical}. It is likely that bats used both strategies during their navigation: the spatial memory strategy stores a coarse-grained representation of the environment and achieves long-term goals, such as arriving at a feeding ground, while the cue-directed strategy stores information of sensory features of the landmarks (nodes of the spatial representation) and achieves short-term goals, such as arriving near a landmark or avoiding obstacles. 

We turn to the question of how these navigation strategies can be interpreted and understood in terms of the sensorimotor laws proposed above.  If there are synthesized trajectories that are well matched to the observational data, these may provide evidence to support the animals' use of one or the other navigation strategies. We thus seek to synthesize trajectories under two different hypotheses:
\begin{enumerate}
\item \textit{Cue-directed Strategy} With this strategy, bats generate their control based purely on sensory information about environmental features (tree trunks and branches and the two obstacles [vine, pole]in our case). This strategy can be understood as a stimulus-response type control. \label{hypothesis_1}
\item \textit{Integrated Strategy} With this strategy, bats use landmarks specified by their spatial memory to select features from their environment and then generate control based on these filtered features. This strategy can be understood as a hierarchical control with landmarks and spatial memory specifying sub-goals or switching points and local sensory features regulating lower-level control.\label{hypothesis_2}
\end{enumerate}

\begin{figure}[!tbh]
\centering
\includegraphics[width=.8\columnwidth]{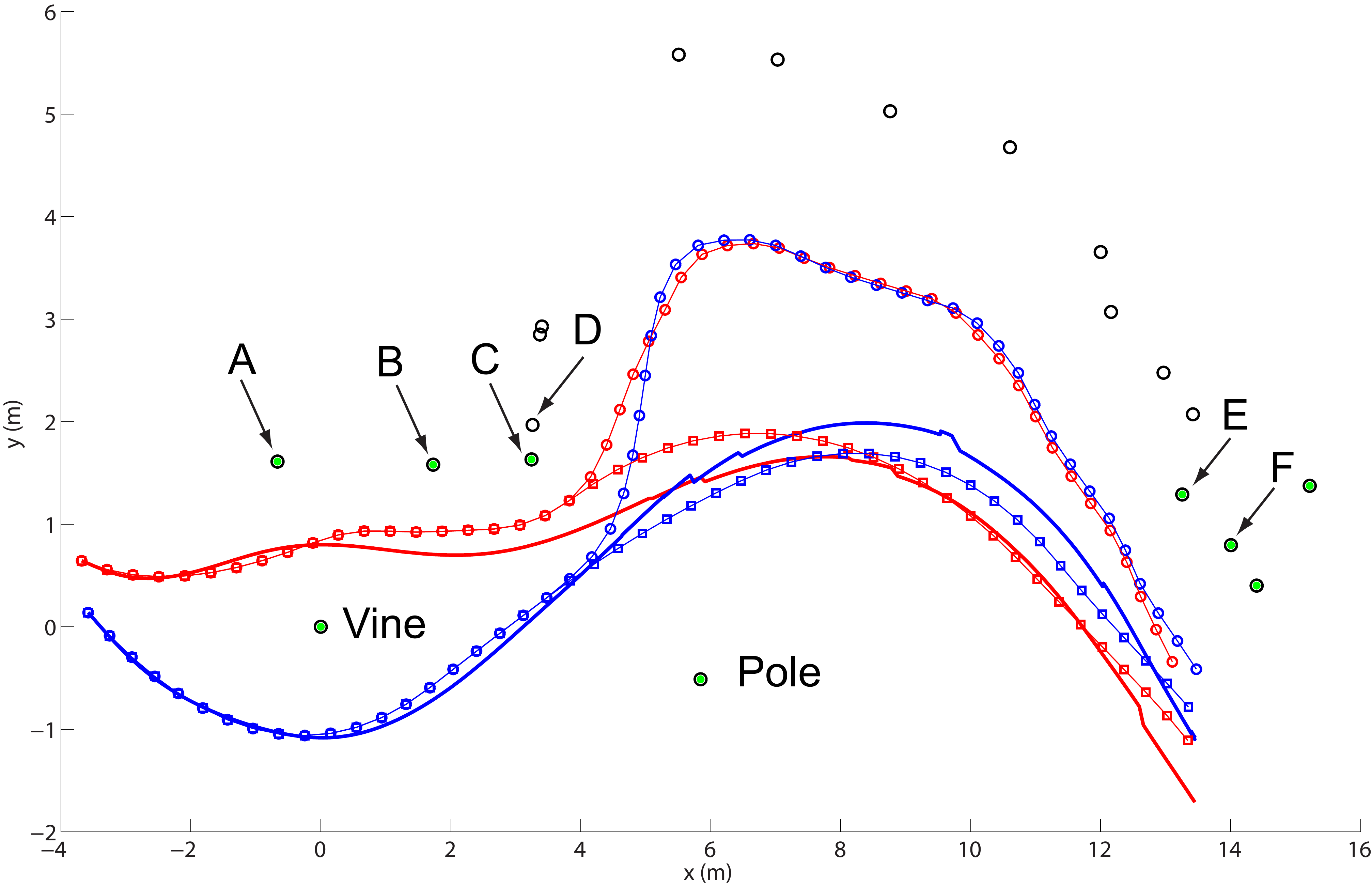}
\caption{Mean bat trajectories (solid curves), synthesized trajectories based on Hypothesis \ref{hypothesis_1} (curves with circles) and synthesized trajectories based on Hypothesis \ref{hypothesis_2} (curves connecting squares).  Features used for navigation under Hypothesis \ref{hypothesis_1} are marked by black circles. The subset of features used under Hypothesis \ref{hypothesis_2} are marked by circles with green interior. Red and blue curves correspond to groups of bats passing the vine and the left and the right, respectively.} \label{fig:Spatial_Memory}
\end{figure}

Fig.~\ref{fig:Spatial_Memory} shows the synthesized trajectories based on these two hypotheses as well as the mean bat trajectories for groups of bats passing the vine from the left and the right. Three types of control laws from Table 1 are used in the syntheses of these trajectories: the distance maintenance control law $u_d[{\cal O}_1,{\cal O}_2]$, the circling control law $u_c[{\cal O}_1]$ and the safe passing control law $u_p[{\cal O}_1,{\cal O}_2]$, where ${\cal O}_1$ and ${\cal O}_2$ are the features used in a particular control law. Notice that as stated before $u_d$ and $u_p$ act on feature pairs while $u_c$ acts on a single feature. The four trajectories are generated by the following sequences of controlled motion segments:  (Feature labels $A,B,\dots$ refer to Fig.~\ref{fig:Spatial_Memory}.
\begin{enumerate}
\item Red curve with squares: $u_p[A, \text{vine}] \rightarrow u_d[A,B] \rightarrow u_d[B,C] \rightarrow u_c[\text{pole}] \rightarrow u_d[E,F] \rightarrow u_d[\cdot,\cdot]$ for the remaining features;
\item Blue curve with squares: $u_c[\text{vine}] \rightarrow u_d[B,C] \rightarrow u_c[\text{pole}]\rightarrow u_d[E,F] \rightarrow u_d[\cdot,\cdot]$ for the remaining features;
\item Red curve with circles: $u_p[A, \text{vine}] \rightarrow u_d[A,B] \rightarrow u_d[B,C] \rightarrow u_c[C] \rightarrow u_d[C,D] \rightarrow u_d[\cdot,\cdot]$ for the remaining features;
\item Blue curve with circles: $u_c[\text{vine}] \rightarrow u_d[B,C] \rightarrow u_c[C] \rightarrow u_d[C,D] \rightarrow u_d[\cdot,\cdot]$ for the remaining features.
\end{enumerate}

As shown by the relatively high rise of the circle-interpolated curves in Fig.~\ref{fig:Spatial_Memory}, if a cue-directed strategy is adopted by the bats, after passing point C, both of these two groups would keep following the edge of the woods. On the other hand, if an integrated strategy is adopted by the bats, after passing point C, the bats can use their spatial memory or the memorized sensory features of significant landmarks (such as the trees marked as E and F) for navigation and then use these features to generate their control. Using the integrated strategy, instead of continuing to follow the border of the woods, the bats would take a short cut. By comparing the synthesized trajectories based on these two hypotheses and the mean bat trajectories, it can be found that the trajectories based on the integrated view fit the observed mean trajectories better. Such an observation suggests that the bats indeed adopt an integrated strategy for their navigation. (A pure spatial memory strategy can be rejected by the observation that none of the bats collided with the pole, which was placed there by the researchers during the experiment and unlikely memorized by the bats.)

An advantage of the integrated strategy over the cue-directed strategy is its time and energy efficiency. By taking the short cut (optimized route) instead of strictly following the forest edge, the bats can save flight time. In addition, by developing a stereotypical flight route---which is observed in our data and other experiments \cite{thiele2005hierarchical}
---and efficient sensorimotor coordination, the bats can also save energy. Flight is an energetically costly mode of locomotion (\cite{Thomas}). In addition to the metabolic demands of flight, many of the individuals roosting in these summer colonies during the time of filming were lactating females that are allocating significant energy resources to parental care (\cite{Fitch,Kurata}). Thus, a behavior that limits energy expenditure by exploiting shortcuts should be favored, even though these individuals are more susceptible to predation by aerial carnivores. 
It is important to note, however, that dealing only with mean behavior may give a misleading picture of the characteristics displayed by the animals.  They exhibit considerable diversity in the way they react to features in the environment, and this is noted in Fig.~\ref{fig:jb:HypothesisTesting}.  Clearly a significant percentage of the paths is consistent with the hypothesis that they are replying on perceived features to navigate.

\begin{figure}[!tbh]
\centering
\includegraphics[width=.8\columnwidth]{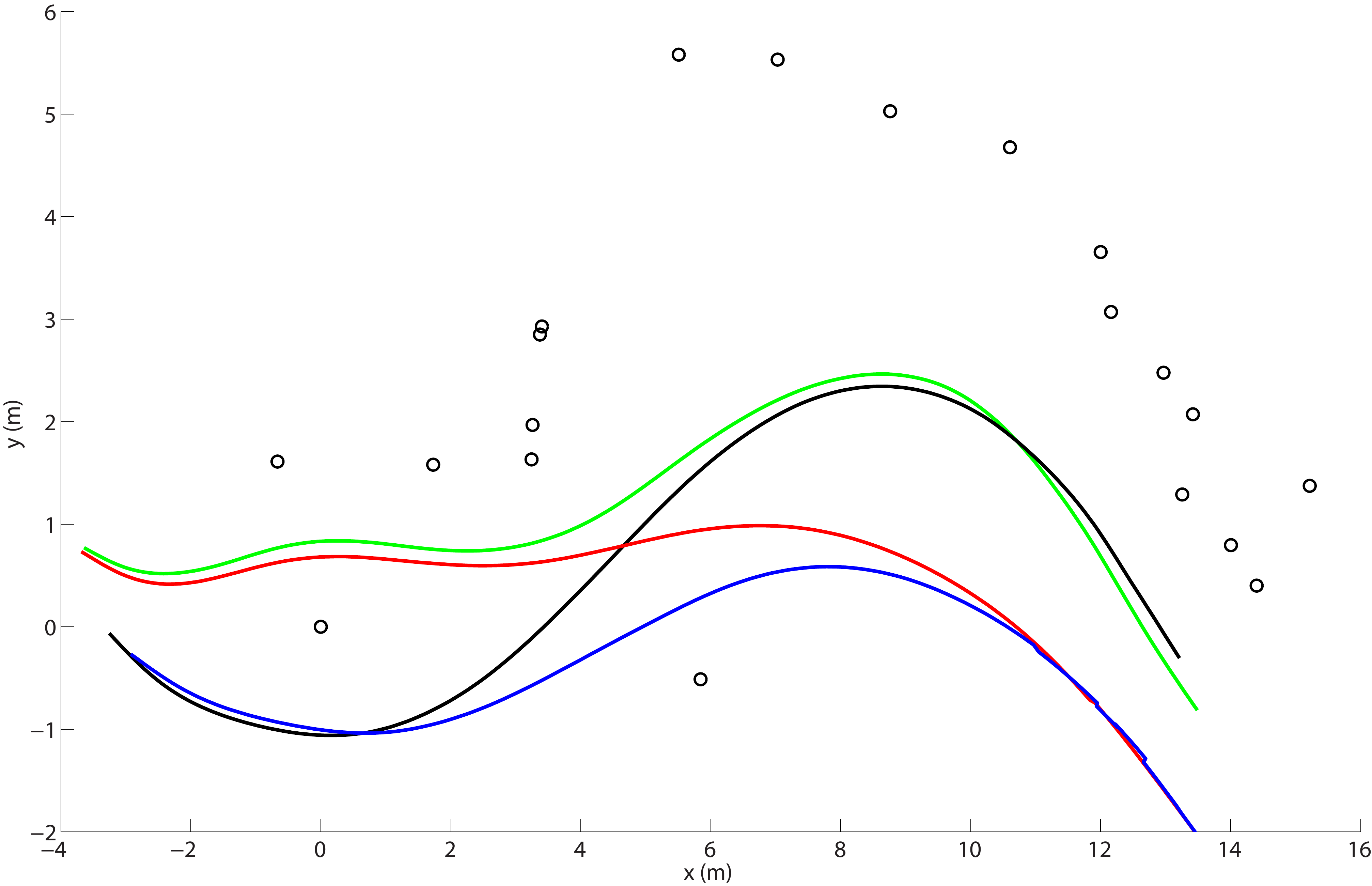}
\caption{Neither hypothesis regarding cue-based versus memory-based navigation fit all bat trajectories.  While the mean trajectories in Fig.~\ref{fig:Spatial_Memory} do not exhibit much of an excursion into the woods, we display here the mean paths of the 35\% of the bats that penetrate more than two meters beyond the vine into the woods (black and green) versus the mean paths of the 65\% that do not venture that far (red and blue).} \label{fig:jb:HypothesisTesting}
\end{figure}

\section{Discussion and Summary}
\label{sec:jb:Discussion}

This paper has given a snapshot summary of an interdisciplinary research effort by a team of biologists, computer scientists, and control theorists aimed at understanding the flight behaviors of bats---more specifically {\em M.~velfier}.  The aim of the research has been twofold.  One goal is to develop principles of navigation and flight control for UAVs that exhibit the agility and efficiency of animal flyers.  An important second goal is  to use synthesized flight behaviors to make inferences about the mechanisms of sensorimotor control employed by the bats to move through their habitat.  
Section \ref{sec:jb:Simulations} reports preliminary results on analysis and interpretation of flight behaviors created by sequences of simple vision-based steering laws.
Due to the limited space allowed in this  paper, we have not been able to discuss many results of the research---including a detailed statistical characterization of the geometric diversity of the bat trajectories described in Section \ref{sec:jb:BatData}.  These results will be presented elsewhere.  Our current work is aimed at developing a richer set of control primitives that will expand the set displayed in Table 1 and allow enhanced parametric studies of flight behaviors.  The work reported has been primarily concerned with vision-based navigation, but future research will study how species such as {\em M.~velifer} might implement flight control that is primarily based on echolocation.  Research is being planned to observe bats in flight domains where evidence for preferring one or the other sensing modality can be obtained together with a deeper understanding of how these sensory modalities are integrated.


\end{document}